\documentclass[11pt]{article}
\usepackage{epsf}
\usepackage{amsmath}
\usepackage{epsfig}
\usepackage{times}
\usepackage{amssymb}
\usepackage{amsthm}
\usepackage{setspace}
\usepackage{cite}

\usepackage{algorithmic}  
\usepackage{algorithm}

\usepackage{shadow}
\usepackage{fancybox}
\usepackage{fancyhdr}

\def\w{{\bf w}}

\def\y{{\bf y}}

\def\x{{\bf x}}

\def\x{{\mathbf x}}

\def\w{{\bf w}}

\def\x{{\bf x}}
\def\y{{\bf y}}
\def\z{{\bf z}}
\def\q{{\bf q}}
\def\a{{\bf a}}
\def\b{{\bf b}}

\def\h{{\bf h}}

\def\be{\begin{equation}}
\def\ee{\end{equation}}
\def\ba{\left[\begin{array}}
\def\ea{\end{array}\right]}

\def\t{{\bf t}}

\def\w{{\bf w}}

\def\x{{\bf x}}
\def\y{{\bf y}}
\def\z{{\bf z}}
\def\q{{\bf q}}
\def\a{{\bf a}}
\def\b{{\bf b}}

\def\1{{\bf 1}}

\def\g{{\bf g}}
\def\0{{\bf 0}}

\def\hatx{{\hat{\x}}}

\def\erfinv{\mbox{erfinv}}

\def\betasec{\beta_{sec}}

\def\hthetasec{\hat{\theta}_{sec}}
\def\thetasec{\theta_{sec}}

\newtheorem{theorem}{Theorem}

\begin{document}

\begin{singlespace}

\title {Upper-bounding $\ell_1$-optimization sectional thresholds
\footnote{ This work was supported in part by NSF grant \#CCF-1217857.}
}
\author{
\textsc{Mihailo Stojnic}
\\
\\
{School of Industrial Engineering}\\
{Purdue University, West Lafayette, IN 47907} \\
{e-mail: {\tt mstojnic@purdue.edu}} }
\date{}
\maketitle

\centerline{{\bf Abstract}} \vspace*{0.1in}

In this paper we look at a particular problem related to under-determined linear systems of equations with sparse solutions. $\ell_1$-minimization is a fairly successful polynomial technique that can in certain statistical scenarios find sparse enough solutions of such systems. Barriers of $\ell_1$ performance are typically referred to as its thresholds. Depending if one is interested in a typical or worst case behavior one then distinguishes between the \emph{weak} thresholds that relate to a typical behavior on one side and the \emph{sectional} and \emph{strong} thresholds that relate to the worst case behavior on the other side. Starting with seminal works \cite{CRT,DonohoPol,DOnoho06CS} a substantial progress has been achieved in theoretical characterization of $\ell_1$-minimization statistical thresholds. More precisely, \cite{CRT,DOnoho06CS} presented for the first time linear lower bounds on all of these thresholds. Donoho's work \cite{DonohoPol} (and our own \cite{StojnicCSetam09,StojnicUpper10}) went a bit further and essentially settled the $\ell_1$'s \emph{weak} thresholds. At the same time they also provided fairly good lower bounds on the values on the \emph{sectional} and \emph{strong} thresholds. In this paper, we revisit the \emph{sectional} thresholds and present a simple mechanism that can be used to create solid upper bounds as well. The method we present relies on a seemingly simple but substantial progress we made in studying Hopfield models in \cite{StojnicHopBnds10}.

\vspace*{0.25in} \noindent {\bf Index Terms: Linear systems of equations;
$\ell_1$-optimization; compressed sensing} .

\end{singlespace}

\section{Introduction}
\label{sec:back}

We start by giving a brief overview of the problem at hand and what we consider as the most relevant mathematical results. In this paper we will be interested in mathematical studying of a particular problem related to under-determined systems of linear equations with sparse solutions. We start by looking at the following system of linear equations
\begin{equation}
A\x=\y, \label{eq:system}
\end{equation}
where $A$ is an $m\times n$ ($m<n$) matrix and $\y$ is
an $m\times 1$ vector. Clearly, as in any linear system the goal is to determine $\x$ if $A$ and $\y$ are given. Given the above dimensions this system is obviously under-determined and for given $A$ and $\y$ the odds are that it will have an infinite number of solutions. In this paper we will be interested in a particular subclass of these systems, namely the one where $y$ is such that (\ref{eq:system}) is satisfied for a $k$-sparse $\x$ and at the same time is not satisfied for any $\x$ that is less than $k$-sparse (here and in the rest of the paper, under $k$-sparse vector we assume a vector that has at most $k$ nonzero
components).

To make writing in the rest of the paper easier, we will assume the
so-called \emph{linear} regime, i.e. we will assume that $k=\beta n$
and that the number of equations is $m=\alpha n$ where
$\alpha$ and $\beta$ are constants independent of $n$ (more
on the non-linear regime, i.e. on the regime when $m$ is larger than
linearly proportional to $k$ can be found in e.g.
\cite{CoMu05,GiStTrVe06,GiStTrVe07}).

There are of course many ways how one can attempt to recover $\x$ in (\ref{eq:system}). Here we only mention a few that are applicable for any matrix $A$.

Typically, the following two algorithms (and their different
variations) have been often viewed historically as solid heuristics for solving (\ref{eq:system}) (in recent years belief propagation type of algorithms are emerging as strong alternatives as well):
\begin{enumerate}
\item \underline{\emph{Orthogonal matching pursuit - OMP}}
\item \underline{\emph{Basis pursuit -
$\ell_1$-optimization.}}
\end{enumerate}
Under certain probabilistic assumptions on the elements of $A$ it can be shown (see e.g. \cite{JATGomp,JAT,NeVe07})
that if $m=O(k\log(n))$
OMP (or slightly modified OMP) can recover $\x$ in (\ref{eq:system})
with complexity of recovery $O(n^2)$. On the other hand a stage-wise
OMP from \cite{DTDSomp} recovers $\x$ in (\ref{eq:system}) with
complexity of recovery $O(n \log n)$. Somewhere in between OMP and BP are recent improvements CoSAMP (see e.g. \cite{NT08}) and Subspace pursuit (see e.g. \cite{DaiMil08}), which guarantee (assuming the linear regime) that the $k$-sparse $\x$ in (\ref{eq:system}) can be recovered in polynomial time with $m=O(k)$ equations. Of course, various other techniques are possible and for that matter have been developed in recent years. However, since this paper is mostly concern with a success of a particular technique we refrain from reviewing further algorithms developed for solving (\ref{eq:system}) and defer that to survey type of papers.

Our interest in this paper is the performance of a technique called $\ell_1$-optimization. (Variations of the standard $\ell_1$-optimization from e.g.
\cite{CWBreweighted,SChretien08,SaZh08}) as well as those from \cite{SCY08,FL08,GN03,GN04,GN07,DG08} related to $\ell_q$-optimization, $0<q<1$
are possible as well.) Basic $\ell_1$-optimization algorithm finds $\x$ in
(\ref{eq:system}) by solving the following $\ell_1$-norm minimization problem
\begin{eqnarray}
\mbox{min} & & \|\x\|_{1}\nonumber \\
\mbox{subject to} & & A\x=\y. \label{eq:l1}
\end{eqnarray}
Due to its popularity the literature on the use of the above algorithm is rapidly growing. We below restrict our attention to two, in our mind, the most influential works that relate to (\ref{eq:l1}).

The first one is \cite{CRT} where the authors were able to show that if
$\alpha$ and $n$ are given, $A$ is given and satisfies the restricted isometry property (RIP) (more on this property the interested reader can find in e.g. \cite{Crip,CRT,Bar,Ver,ALPTJ09}), then
any unknown vector $\x$ with no more than $k=\beta n$ (where $\beta$
is a constant dependent on $\alpha$ and explicitly
calculated in \cite{CRT}) non-zero elements can be recovered by
solving (\ref{eq:l1}).

However, the RIP is only a \emph{sufficient}
condition for $\ell_1$-optimization to produce the $k$-sparse solution of
(\ref{eq:system}). Instead of characterizing $A$ through the RIP
condition, in \cite{DonohoUnsigned,DonohoPol} Donoho looked at its geometric properties/potential. Namely,
in \cite{DonohoUnsigned,DonohoPol} Donoho considered polytope obtained by
projecting the regular $n$-dimensional cross-polytope $C_p^n$ by $A$. He then established that
the solution of (\ref{eq:l1}) will be the $k$-sparse solution of
(\ref{eq:system}) if and only if
$AC_p^n$ is centrally $k$-neighborly
(for the definitions of neighborliness, details of Donoho's approach, and related results the interested reader can consult now already classic references \cite{DonohoUnsigned,DonohoPol,DonohoSigned,DT}). In a nutshell, using the results
of \cite{PMM,AS,BorockyHenk,Ruben,VS}, it is shown in
\cite{DonohoPol}, that if $A$ is a random $m\times n$
ortho-projector matrix then with overwhelming probability $AC_p^n$ is centrally $k$-neighborly (as usual, under overwhelming probability we in this paper assume
a probability that is no more than a number exponentially decaying in $n$ away from $1$). Miraculously, \cite{DonohoPol,DonohoUnsigned} provided a precise characterization of $m$ and $k$ (in a large dimensional context) for which this happens.

It should be noted that one usually considers success of
(\ref{eq:l1}) in recovering \emph{any} given $k$-sparse $\x$ in (\ref{eq:system}). It is also of interest to consider success of
(\ref{eq:l1}) in recovering
\emph{almost any} given $\x$ in (\ref{eq:system}). We below make a distinction between these
cases and recall on some of the definitions from
\cite{DonohoPol,DT,DTciss,DTjams2010,StojnicCSetam09,StojnicICASSP09}.

Clearly, for any given constant $\alpha\leq 1$ there is a maximum
allowable value of $\beta$ such that for \emph{any} given $k$-sparse $\x$ in (\ref{eq:system}) the solution of (\ref{eq:l1})
is with overwhelming probability exactly that given $k$-sparse $\x$. We will refer to this maximum allowable value of
$\beta$ as the \emph{strong} threshold (see
\cite{DonohoPol}). Similarly, for any given constant
$\alpha\leq 1$ and \emph{any} given $\x$ with a given fixed location of non-zero components
there will be a maximum allowable value of $\beta$ such that
(\ref{eq:l1}) finds that given $\x$ in (\ref{eq:system}) with overwhelming
probability. We will refer to this maximum allowable value of
$\beta$ as the \emph{sectional} threshold and will denote it by $\beta_{w}$
Finally, for any given constant
$\alpha\leq 1$ and \emph{any} given $\x$ with a given fixed location of non-zero components and a given fixed combination of its elements signs
there will be a maximum allowable value of $\beta$ such that
(\ref{eq:l1}) finds that given $\x$ in (\ref{eq:system}) with overwhelming
probability. We will refer to this maximum allowable value of
$\beta$ as the \emph{weak} threshold and will denote it by $\beta_{w}$ (see, e.g. \cite{StojnicICASSP09,StojnicCSetam09}).

When viewed within this frame the results of \cite{CRT,DOnoho06CS} established that $\ell_1$-minimization achieves recovery through a linear scaling of all important dimensions ($k$, $m$, and $n$). Moreover, for all $\beta$'s defined above lower bounds were provided in \cite{CRT}. On the other hand, the results of \cite{DonohoPol,DonohoUnsigned} established the exact values of $\beta_w$ and provided lower bounds on $\beta_{str}$ and $\beta_{sec}$.

In a series of our own work (see, e.g. \cite{StojnicICASSP09,StojnicCSetam09,StojnicUpper10}) we then created an alternative probabilistic approach which was capable of providing the precise characterization of $\beta_w$ as well and thereby of reestablishing the results of Donoho \cite{DonohoPol} through a purely probabilistic approach. We also presented in \cite{StojnicCSetam09} further results related to lower bounds on $\beta_{str}$ and $\beta_{sec}$.

Our main subject of interest in this paper is the \emph{sectional} threshold. Before proceeding further with the presentation we find it useful to restate the results from \cite{StojnicCSetam09} that relate to the sectional thresholds $\beta_{sec}$. The following theorem summarizes these results. We will fairly often use the results of this theorem as a sort of benchmark for the results that we will present in this paper.

\begin{theorem}(Sectional threshold - lower bound)
Let $A$ be an $m\times n$ measurement matrix in (\ref{eq:system})
with the null-space uniformly distributed in the Grassmanian. Let
the unknown $\x$ in (\ref{eq:system}) be $k$-sparse. Further, let the location of nonzero elements of $\x$ be arbitrarily chosen but fixed.
Let $k,m,n$ be large
and let $\alpha=\frac{m}{n}$ and $\betasec=\frac{k}{n}$ be constants
independent of $m$ and $n$. Let $\erfinv$ be the inverse of the standard error function associated with zero-mean unit variance Gaussian random variable.  Further,
let $\epsilon>0$ be an arbitrarily small constant and $\hthetasec$, ($\betasec\leq \hthetasec\leq 1$) be the solution of
\begin{equation}
(1-\epsilon)(1-\betasec)\frac{\sqrt{\frac{2}{\pi}}e^{-(\erfinv(\frac{1-\thetasec}{1-\betasec}))^2}-\sqrt{\frac{2}{\pi}}\frac{\betasec}{1-\betasec}}
{\thetasec}-\sqrt{2}\erfinv ((1+\epsilon)\frac{1-\thetasec}{1-\betasec})=0.\label{eq:thmsectheta}
\end{equation}
If $\alpha$ and $\betasec$ further satisfy
\begin{equation}
\hspace{-.6in}\alpha>\frac{1-\betasec}{\sqrt{2\pi}}\left (\sqrt{2\pi}+2\frac{\sqrt{2(\erfinv(\frac{1-\hthetasec}{1-\betasec}))^2}}{e^{(\erfinv(\frac{1-\hthetasec}{1-\betasec}))^2}}-\sqrt{2\pi}
\frac{1-\hthetasec}{1-\betasec}\right )+\betasec
-\frac{\left ((1-\betasec)\sqrt{\frac{2}{\pi}}e^{-(\erfinv(\frac{1-\hthetasec}{1-\betasec}))^2}-\sqrt{\frac{2}{\pi}}\betasec\right )^2}{\hthetasec}\label{eq:thmsecalpha}
\end{equation}
then with overwhelming probability the solution of (\ref{eq:l1}) is the $k$-sparse $\x$ from (\ref{eq:system}).
\label{thm:thmsecthr}
\end{theorem}

The above theorem was obtained in \cite{StojnicCSetam09} through a novel probabilistic framework for performance characterization of (\ref{eq:l1}). Using that framework we obtained lower bounds on $\beta_{sec}$. These lower bounds are not exact. In this paper we design a mechanism that can be used to compute the upper bounds on $\beta_{sec}$. The obtained upper bounds will obviously not match the lower bounds computed in \cite{StojnicCSetam09} but are relatively simple to compute and can provide a quick assessment as to how far off from the optimal are in the worst the results obtained for sectional thresholds in \cite{StojnicCSetam09}.

Although studying the weak thresholds is not the subject of this paper, we should as a side point mention that the weak thresholds computed in \cite{StojnicCSetam09} were confirmed in \cite{StojnicUpper10,StojnicEquiv10} to be the exact ones. In this paper we will also utilize to a degree the upper-bounding methodology of \cite{StojnicUpper10}. However, a few further insights are needed to make the mechanism we are about to present work and those became available only after we made a simple but important progress in studying a class of Hopfield models from statistical physics in \cite{StojnicHopBnds10}.

We organize the rest of the paper in the following way. In Section
\ref{sec:unsigned} we create a mechanism for computing the upper bounds on $\beta_{sec}$ for a class of random matrices $A$. In Section \ref{sec:simul} we present a collection of numerical results that aim at estimating how far off are our upper bounds from true $\beta_{sec}$. Finally, in Section \ref{sec:discuss} we discuss obtained results.

\section{Upper-bounding $\beta_{sec}$} \label{sec:unsigned}

In this section we present the mechanism for upper-bounding the sectional thresholds. We first recall on a sectional type of optimality characterization of (\ref{eq:l1}). Such a characterization is completely deterministic. We in the second part of this section then probabilistically analyze the obtained characterization.

\subsection{Deterministic part} \label{sec:det}

Namely, we look at a null-space characterization of
$A$ that guarantees (in a sectional sense) that the solution of
(\ref{eq:l1}) is the $k$-sparse solution of (\ref{eq:system}). To be more precise, the characterization will establish a condition which is equivalent to having the solution of
(\ref{eq:l1}) be the $k$-sparse solution of (\ref{eq:system}) for any $\beta n$-sparse $\x$ with a fixed location of nonzero components. Since the analysis will clearly be irrelevant with respect to what particular location is chosen, we can for the simplicity of the exposition and without loss of generality assume that the components $\x_{1},\x_{2},\dots,\x_{n-k}$ of $\x$ are equal to zero and the components $\x_{n-k+1},\x_{n-k+2},\dots,\x_n$ of $\x$ are larger than or equal to zero. Under this assumption we have the following theorem from \cite{StojnicICASSP09} that provides
such a characterization (while the corresponding weak threshold characterization was introduced for the first time in \cite{StojnicICASSP09}, the sectional characterization we need here was by no means derived in \cite{StojnicICASSP09} for the first time; similar sectional/strong threshold characterizations were obtained way earlier, see e.g.
\cite{DH01,FN,LN,Y,XHapp,SPH,DTbern}; furthermore, if instead of $\ell_1$ one, for
example, uses an $\ell_q$-optimization ($0<q<1$) in (\ref{eq:l1}) then
characterizations similar to the ones from
\cite{DH01,FN,LN,Y,XHapp,SPH,DTbern} can be derived as well
\cite{GN03,GN04,GN07}).
\begin{theorem}(Nonzero part of $\x$ has fixed location)
Assume that an $m\times n$ matrix $A$ is given. Let $\x$
be a $k$-sparse vector. Also let $\x_1=\x_2=\dots=\x_{n-k}=0.$
Further, assume that $\y=A\x$ and that $\w$ is
an $n\times 1$ vector. If
\begin{equation}
(\forall \w\in \textbf{R}^n | A\w=0) \quad  \sum_{i=n-k+1}^n |\w_i|<\sum_{i=1}^{n-k}|\w_{i}|
\label{eq:thmeqgenweak1}
\end{equation}
then the solution of (\ref{eq:l1}) is $\x$. Moreover, if
\begin{equation}
(\exists \w\in \textbf{R}^n | A\w=0) \quad  \sum_{i=n-k+1}^n |\w_i|>\sum_{i=1}^{n-k}|\w_{i}|
\label{eq:thmeqgenweak2}
\end{equation}
then there will be a $k$-sparse $\x$ that satisfies (\ref{eq:system}) and is not the solution of (\ref{eq:l1}).
\label{thm:thmgenweak}
\end{theorem}
\begin{proof}
The first part follows directly from Theorem $2$ in \cite{StojnicICASSP09} by viewing a particular subset of locations. For the completeness we just sketch the argument again. Let $\hatx$ be the solution of (\ref{eq:l1}). We want to show that if (\ref{eq:thmeqgenweak1}) holds then $\hatx=\x$. To that end assume opposite, i.e. assume that (\ref{eq:thmeqgenweak1}) holds but $\hatx\neq\x$. Then since $\y=A\hatx$ and $\y=A\x$ one must have $\hatx =\x+\w$ with $\w$ such that $A\w=0$. Also, since $\hatx$ is the solution of (\ref{eq:l1}) one has that
\begin{equation}
\sum_{i=1}^n|\x_i+\w_i|\leq \sum_{i=1}^{n}|\x_i|.\label{eq:absval}
\end{equation}
Then the following must hold as well
\begin{equation}
\sum_{i=1}^{n-k} |\w_i|-\sum_{i=n-k+1}^n |\w_i|\leq 0.\label{eq:wcon}
\end{equation}
or equivalently
\begin{equation}
\sum_{i=1}^{n-k} |\w_i|\leq\sum_{i=n-k+1}^n |\w_i|.\label{eq:wcon1}
\end{equation}
Clearly, (\ref{eq:wcon1}) contradicts (\ref{eq:thmeqgenweak1}) and $\hatx\neq\x$ can not hold. Therefore $\hatx=\x$ which is exactly what the first part of the theorem claims.

For the ``moreover" part assume that (\ref{eq:thmeqgenweak2}) holds, i.e. we assume
\begin{equation}
(\exists \w\in \textbf{R}^n | A\w=0) \quad  \sum_{i=n-k+1}^n \w_i>\sum_{i=1}^{n-k}|\w_{i}|\label{eq:thmeqgenweak3}
\end{equation}
and want to show that there is a $k$-sparse $\x$ with $\x_1=\x_2=\dots=\x_{n-k}=0$ such that (\ref{eq:absval}) holds (with a strict inequality). This would imply that there is a $\x$ with $\x_1=\x_2=\dots=\x_{n-k}=0$ such that $A\x=\y$ and $\x$ is not the solution of (\ref{eq:l1}). Since (\ref{eq:wcon1}) is just rewritten (\ref{eq:thmeqgenweak3}) one can go backwards from (\ref{eq:wcon1}) to (\ref{eq:absval}) (just additionally making all the inequalities strict in the process). Then for $\x$ such that $\x_j=0$ for $1\leq j\leq n-k$,  $\x_j=-\w_j,n-k+1\leq j\leq n$ one has that (\ref{eq:thmeqgenweak3}) implies
\begin{equation}
\sum_{i=1}^n|\x_i+\w_i|< \sum_{i=1}^{n}|\x_i|.\label{eq:absvalrev}
\end{equation}
or in other words that $\x$ can not be the solution of (\ref{eq:l1}). This concludes the proof of the second (``moreover") part.
\end{proof}

We believe that a few comments are in order. Clearly, the first part of the above theorem is the characterization that was used to obtain the lower bounds on the sectional thresholds in \cite{StojnicCSetam09} (and way earlier in \cite{DonohoPol}). The second part may seem somewhat novel when it comes to its use in sectional thresholds characterizations. However, we should emphasize that its statement and proof are nothing original (see, e.g. \cite{DH01,GN03}). On the other hand, as mentioned above we have hardly ever seen any use of the second part before. Of course that is somewhat expected as long as one is concerned with the lower bounds. However, as the reader might guess, if one is concerned with proving the upper bounds the second part of the above theorem becomes the same type of the key proving strategy component that the first part was in the framework of \cite{StojnicCSetam09}. Below we use it to create a machinery almost as powerful as the one from \cite{StojnicCSetam09} that provides the corresponding framework for upper-bounding the sectional thresholds.

\subsection{Probabilistic part} \label{sec:prob}

In this section we probabilistically analyze validity of the null-space characterization given in the second part of Theorem \ref{thm:thmgenweak}. Essentially, we will design a mechanism for computing upper bounds on $\beta_{sec}$ (in fact, since it will be slightly more convenient we will actually determine lower bounds on $\alpha$; that is of course conceptually  the same as finding the upper-bounds on $\beta_{sec}$). In the first part of this subsection we will closely follow the strategy presented in \cite{StojnicUpper10} used to obtain upper bounds on the weak thresholds.

We start by defining a quantity $\tau$ that will play one of the key roles below
\begin{eqnarray}
\tau(A) =  \min & & (\sum_{i=1}^{n-k}|\w_i|-\sum_{i=n-k+1}^{n}|\w_i|) \nonumber \\
\mbox{subject to} & &  A\w=0\nonumber \\
& & \|\w\|_2\leq 1.\label{eq:deftau}
\end{eqnarray}
Now, we will in the rest of the paper assume that the entries of $A$ i.i.d. standard normal random variables. Then one can say that for any $\alpha$ and $\beta$ for which
\begin{equation}
\lim_{n\rightarrow\infty}P(\tau(A)<0)=1, \label{eq:mcritpr}
\end{equation}
there is a $k$-sparse $\x$ (from a set of $\x$'s with a given fixed location of nonzero components) which (\ref{eq:l1}) with probability $1$ fails to find. For a fixed $\beta$ our goal will be to find the largest possible $\alpha$ for which (\ref{eq:mcrit}) holds, i.e. for which (\ref{eq:l1}) fails with probability $1$. As is now well known based on the machinery developed in a series of our work \cite{StojnicCSetam09,StojnicUpper10} all random quantities of interest will concentrate and one can instead of looking at (\ref{eq:mcrit}) look at the alternative condition
\begin{equation}
\lim_{n\rightarrow\infty}\frac{E\tau(A)}{\sqrt{n}}<0. \label{eq:mcrit}
\end{equation}
Before going through the randomness of the problem and evaluation of $\lim_{n\rightarrow\infty}\frac{E\tau(A)}{\sqrt{n}}$ (and ultimately $P(\tau(A)<0)$) we will try to provide a more explicit expression for $\tau$ than the one given by the optimization problem in (\ref{eq:deftau}). We proceed by slightly rephrasing (\ref{eq:deftau}):
\begin{eqnarray}
\tau(A) = \min_{\b_i^2=1} \min_{\t,\w} & & (\sum_{i=1}^{n-k}\t_i-\sum_{i=n-k+1}^{n}\b_i\w_i) \nonumber \\
\mbox{subject to} & & -\t_i\leq \w_i\leq \t_i, 1\leq i\leq n-k \nonumber \\
& &  A\w=0\nonumber \\
& & \|\w\|_2\leq 1.\label{eq:deftau1}
\end{eqnarray}
We then write further
\begin{equation}
\tau(A) = \min_{\b_i^2=1} \tau_w(A,\b),\label{eq:deftau11}
\end{equation}
where
\begin{eqnarray}
\tau_w(A,\b) = \min_{\t,\w} & & (\sum_{i=1}^{n-k}\t_i-\sum_{i=n-k+1}^{n}\b_i\w_i) \nonumber \\
\mbox{subject to} & & -\t_i\leq \w_i\leq \t_i, 1\leq i\leq n-k \nonumber \\
& &  A\w=0\nonumber \\
& & \|\w\|_2\leq 1.\label{eq:deftau12}
\end{eqnarray}
Now, one can closely follow what was done in \cite{StojnicUpper10} between equations $(14)$ and $(25)$ to arrive to the following analogue of \cite{StojnicUpper10}'s $(25)$
\begin{eqnarray}
\tau_w(A,\b)= \max_{\z,\nu}&  &
-\|\z-A^T \nu\|_2  \nonumber \\
\mbox{subject to}& & |\z_i|\leq 1, 1\leq i\leq n-k\nonumber \\
& & \z_i=-\b_i, n-k+1\leq i\leq n.
\label{eq:deftau11}
\end{eqnarray}
or in a more convenient form
\begin{eqnarray}
\tau_w(A,\b)= -\min_{\z,\nu}&  &
\|\z-A^T \nu\|_2  \nonumber \\
\mbox{subject to}& & |\z_i|\leq 1, 1\leq i\leq n-k\nonumber \\
& & \z_i=-\b_i, n-k+1\leq i\leq n.
\label{eq:deftau12}
\end{eqnarray}
Now, we proceed by solving the inner minimization over $\nu$. To that end we write
\begin{eqnarray}
\tau_w(A,\b)^2= -\min_{\z}\min_{\nu}&  &
\nu^T AA^T\nu-2\z^TA^T\nu +\|\z\|_2^2 \nonumber \\
\mbox{subject to}& & |\z_i|\leq 1, 1\leq i\leq n-k\nonumber \\
& & \z_i=-\b_i, n-k+1\leq i\leq n.
\label{eq:deftau120}
\end{eqnarray}
Since
\begin{eqnarray}
\min_{\nu}
\nu^T AA^T\nu-2\z^TA^T\nu=-\z^TA^T(AA^T)^{-1}A\z,
\label{eq:deftau121}
\end{eqnarray}
one then from (\ref{eq:deftau120}) has
\begin{eqnarray}
\tau_w(A,\b)^2= -\min_{\z}&  &
-\z^TA^T(AA^T)^{-1}A\z+\|\z\|_2^2 \nonumber \\
\mbox{subject to}& & |\z_i|\leq 1, 1\leq i\leq n-k\nonumber \\
& & \z_i=-\b_i, n-k+1\leq i\leq n,
\label{eq:deftau122}
\end{eqnarray}
and alternatively
\begin{eqnarray}
\tau_w(A,\b)^2= -\min_{\z}&  &
\z^T(I-A^T(AA^T)^{-1}A)\z \nonumber \\
\mbox{subject to}& & |\z_i|\leq 1, 1\leq i\leq n-k\nonumber \\
& & \z_i=-\b_i, n-k+1\leq i\leq n.
\label{eq:deftau123}
\end{eqnarray}
Now, we look at the SVD decomposition of $A$
\begin{equation}
A=SVD^T,\label{eq:Asvd}
\end{equation}
where $S$ is an $m\times m$ matrix such that $SS^T=I$, $V$ is a diagonal matrix of singular values of $A$, and $D$ is an $n\times m$ matrix such that $D^TD=I$. Then
\begin{equation}
A^T(AA^T)^{-1}A=DVS^T(SV^2S^T)^{-1}SVD^T,\label{eq:Asvd1}
\end{equation}
and
\begin{equation}
A^T(AA^T)^{-1}A=DD^T.\label{eq:Asvd2}
\end{equation}
Let $D^{\perp}$ be an $(n-m)\times n$ matrix such that
\begin{equation}
\begin{bmatrix}D & D^{\perp}\end{bmatrix}^T\begin{bmatrix}D & D^{\perp}\end{bmatrix}=I.\label{eq:Asvd3}
\end{equation}
Then one also has
\begin{equation}
\begin{bmatrix}D & D^{\perp}\end{bmatrix}\begin{bmatrix}D & D^{\perp}\end{bmatrix}^T=I,\label{eq:Asvd4}
\end{equation}
or in other words
\begin{equation}
I-DD^T = D^{\perp}(D^{\perp})^T.\label{eq:Asvd6}
\end{equation}
Using (\ref{eq:Asvd6}), (\ref{eq:deftau123}) becomes
\begin{eqnarray}
\tau_w(A,\b)^2= -\min_{\z}&  &
\z^T((D^{\perp})^TD^{\perp})\z \nonumber \\
\mbox{subject to}& & |\z_i|\leq 1, 1\leq i\leq n-k\nonumber \\
& & \z_i=-\b_i, n-k+1\leq i\leq n,
\label{eq:deftau124}
\end{eqnarray}
and obviously
\begin{eqnarray}
\tau_w(A,\b)= -\min_{\z}&  &
\|D^{\perp}\z\|_2 \nonumber \\
\mbox{subject to}& & |\z_i|\leq 1, 1\leq i\leq n-k\nonumber \\
& & \z_i=-\b_i, n-k+1\leq i\leq n.
\label{eq:deftau125}
\end{eqnarray}
Given the rotational invariance of Gaussian matrices and the fact that one is ultimately only interested in the sign of $\tau_w(A,\b)$, from a statistical point of view one can then replace $D^{\perp}$ with an $(n-m)\times n$ matrix $A^{(w)}$ with i.i.d. standard normal components. One can then write
\begin{eqnarray}
\tau(A)= -\max_{\b_i^2=1}\min_{\z}&  &
\|D^{\perp}\z\|_2 \nonumber \\
\mbox{subject to}& & |\z_i|\leq 1, 1\leq i\leq n-k\nonumber \\
& & \z_i=-\b_i, n-k+1\leq i\leq n,
\label{eq:deftau126}
\end{eqnarray}
and
\begin{eqnarray}
\tau^{(g)}(A)= -\max_{\b_i^2=1}\min_{\z}&  &
\|A^{(w)}\z\|_2 \nonumber \\
\mbox{subject to}& & |\z_i|\leq 1, 1\leq i\leq n-k\nonumber \\
& & \z_i=-\b_i, n-k+1\leq i\leq n.
\label{eq:deftau127}
\end{eqnarray}
and
\begin{equation}
\lim_{n\rightarrow \infty}\frac{\mbox{sign}(E\tau(A))}{\sqrt{n}}=\lim_{n\rightarrow \infty}\frac{\mbox{sign}(E\tau^{(g)}(A))}{n}.\label{eq:gausstau}
\end{equation}
This essentially means that one can switch to the analysis of the quantity on the right hand side of (\ref{eq:gausstau}). We then have
\begin{eqnarray}
\lim_{n\rightarrow\infty}\frac{E\tau^{(g)}(A)}{n} = \lim_{n\rightarrow\infty}-E\max_{\b\in\{-1,1\}^k}\min_{\z_{1:n-k}}&  & \frac{1}{n}
\|A_{:,n-k+1:n}^{(w)}\b +A_{:,1:n-k}^{(w)}\z_{1:n-k}\|_2 \nonumber \\
\mbox{subject to}& & |\z_i|\leq 1, 1\leq i\leq n-k,
\label{eq:deftau125}
\end{eqnarray}
where
\begin{equation}
\b=[\b_{n-k+1},\b_{n-k+2},\dots,\b_{n}],\label{eq:defvecb}
\end{equation}
$A_{:,1:n-k}^{(w)}$ is a submatrix of $A^{(w)}$ obtained by extracting columns $\{1,2,\dots,n-k\}$, $A_{:,n-k+1:n}^{(w)}$ is a submatrix of $A^{(w)}$ obtained by extracting columns $\{n-k+1,n-k+2,\dots,n\}$, and analogously $\z_{1:n-k}$ is a vector obtained by extracting components $\{1,2,\dots,n-k\}$ of $\z$. From (\ref{eq:deftau125}) we obtain
\begin{eqnarray}
\lim_{n\rightarrow\infty}\frac{E\tau^{(g)}(A)}{n}\leq \lim_{n\rightarrow\infty}-E\min_{\z}&  &
\|(\max_{\b\in\{-1,1\}^k}\|A_{:,n-k+1:n}^{(w)}\b\|_2) \a +A_{:,1:n-k}^{(w)}\z_{1:n-k}\|_2 \nonumber \\
\mbox{subject to}& & |\z_i|\leq 1, 1\leq i\leq n-k,
\label{eq:deftau128}
\end{eqnarray}
where $\a$ is an arbitrary $(n-m)\times 1$ unit norm constant vector. Given statistical independence of columns of $A$ one can first condition on $A_{:,n-k+1:n}^{(w)}$ and set
\begin{equation}
\xi_{n}=\lim_{n\rightarrow \infty}\frac{E(\max_{\b\in\{-1,1\}^k}\|A_{:,n-k+1:n}^{(w)}\b\|_2)}{n}.\label{eq:xihopmax}
\end{equation}
Then from (\ref{eq:deftau128}) we have
\begin{eqnarray}
\lim_{n\rightarrow \infty}\frac{E\tau^{(g)}(A)}{n}\leq \lim_{n\rightarrow\infty}-E\min_{\z}&  & \frac{1}{n}
\|\xi_{n}n \a +A_{:,1:n-k}^{(w)}\z_{1:n-k}\|_2 \nonumber \\
\mbox{subject to}& & |\z_i|\leq 1, 1\leq i\leq n-k.
\label{eq:deftau129}
\end{eqnarray}
Now, if one can indeed compute $\xi_{n}$ we would have a mechanism to establish the condition for negativity of $\lim_{n\rightarrow \infty}\frac{E\tau^{(g)}(A)}{n}$. Computing $\xi_{n}$ is not easy, though. However, following \cite{StojnicHopBnds10} one can design lower and upper bounds on $\xi_{n}$. In fact, here it turns out that the lower bounds are what we need. Using the results of \cite{StojnicHopBnds10} one then has
\begin{equation}
\lim_{n\rightarrow \infty}\frac{k(\sqrt{\frac{1-\alpha}{\beta}}+\xi_{SK})}{n}=\xi_{n}^{(l)}\leq \xi_{n}=\lim_{n\rightarrow \infty}\frac{E(\max_{\b\in\{-1,1\}^k}\|A_{:,n-k+1:n}^{(w)}\b\|_2)}{n},\label{eq:xihopmax1}
\end{equation}
where
\begin{equation}
\xi_{SK}=\lim_{n\rightarrow \infty}\frac{E(\max_{x\in\{-1,1\}^n}\x^T G\x)}{\sqrt{2}\sqrt{n}}\approx 0.7632,\label{eq:defxisk}
\end{equation}
and $G$ is an $n\times n$ matrix with i.i.d. standard normal components.

Then from (\ref{eq:deftau129}) we have
\begin{eqnarray}
\lim_{n\rightarrow \infty}\frac{E\tau^{(g)}(A)}{n}\leq \lim_{n\rightarrow\infty}-E\min_{\z}&  & \frac{1}{n}
\|\xi_{n}^{(l)}n \a +A_{:,1:n-k}^{(w)}\z_{1:n-k}\|_2 \nonumber \\
\mbox{subject to}& & |\z_i|\leq 1, 1\leq i\leq n-k.
\label{eq:deftau1290}
\end{eqnarray}
Before we present the way to handle (\ref{eq:deftau1290}) which is a bit tricky but in our view quite beautiful, we will briefly sketch a standard way how one could proceed based on the mechanisms from \cite{StojnicCSetam09,StojnicUpper10}. Using the mechanisms of \cite{StojnicCSetam09} one can then establish the following upper bound on the right hand side of (\ref{eq:deftau1290}) (in fact, using the machinery of \cite{StojnicCSetam09,StojnicUpper10,StojnicRegRndDlt10} one can actually show that the following upper bound is actually equal to the right hand side of \ref{eq:deftau1290})
\begin{eqnarray}
\lim_{n\rightarrow \infty}\frac{E\tau^{(g)}(A)}{n}\leq \lim_{n\rightarrow\infty} -E\min_{\z}\max_{\q}&  & \frac{1}{n}
(\xi_{n}^{(l)}n \q^T\a +\q^T\g\|\z_{1:n-k}\|_2+\h^T\z_{1:n-k}) \nonumber \\
\mbox{subject to}& & |\z_i|\leq 1, 1\leq i\leq n-k,
\label{eq:deftau1291}
\end{eqnarray}
or alternatively
\begin{eqnarray}
\lim_{n\rightarrow \infty}\frac{E\tau^{(g)}(A)}{n}\leq \lim_{n\rightarrow\infty} -E\min_{\z}&  & \frac{1}{n}
(\sqrt{(\xi_{n}^{(l)}n)^2 +\|\g\|_2^2\|\z_{1:n-k}\|_2^2}+\h^T\z_{1:n-k}) \nonumber \\
\mbox{subject to}& & |\z_i|\leq 1, 1\leq i\leq n-k.
\label{eq:deftau1292}
\end{eqnarray}
The above optimization on the right hand side of the inequality can be solved (it is not that hard but it is a bit involved). In fact, that is how we initially handled (\ref{eq:deftau1290}). However, we then found an alternative way to handle (\ref{eq:deftau1290}) which, as stated above, we consider way more beautiful than the  standard combination of (\ref{eq:deftau1291}) and (\ref{eq:deftau1292}). While we believe that the results should be presented in the original form that tightly follows the way we created them, we simply appreciate the beauty of the alternative method so much that we decided to present that method. Moreover, the form of the final result (although analytically the same as what can be obtained through (\ref{eq:deftau1292})) is way more beautiful in the tricky method that we present below.

Essentially, to handle (\ref{eq:deftau1290}), one can recognize that term $\xi_n^{(l)}n\a$ can, from the statistical point of view, be replaced by $A^{(w,g)}\z^{(g)}$ where $A^{(w,g)}$ is an $(n-m)\times k^{(g)}$ matrix of i.i.d. standard normals (obviously independent of $A^{(w)}$ as well), $\z^{(g)}$ is an $k^{(g)}\times 1$ vector of all $-1$'s ($1$'s work as well; however to make in what follows more obvious the parallel with the results from \cite{StojnicUpper10} $-1$'s work better), and $k^{(g)}$ is
such that
\begin{equation}
k^{(g)}(n-m)=(\xi_n^{(l)})^2n^2.\label{eq:defkg}
\end{equation}
The above condition is obtained from the following line of the identities
\begin{equation}
k^{(g)}(n-m)=E\|A^{(w,g)}\z^{(g)}\|_2^2=(\xi_{n}^{(l)}n)^2\|\a\|_2^2=(\xi_{n}^{(l)})^2n^2.\label{eq:defkg1}
\end{equation}
One can then rewrite (\ref{eq:deftau1290}) as
\begin{eqnarray}
\lim_{n\rightarrow \infty}\frac{E\tau^{(g)}(A)}{n}\leq \lim_{n\rightarrow\infty} -E\min_{\z_{1:n-k},\z^{(g)}}&  & \frac{1}{n}
\|A^{(w,g)}\z^{(g)} +A_{:,1:n-k}^{(w)}\z_{1:n-k}\|_2 \nonumber \\
\mbox{subject to}& & |\z_i|\leq 1, 1\leq i\leq n-k\nonumber\\
& & \z_i^{(g)}=-1, 1\leq i\leq k^{(g)}.
\label{eq:deftau12901}
\end{eqnarray}
(Conditioning on $A^{(w,g)}$ and noting that $\z^{(g)}$ is fixed essentially affirms our above assertion.)
One can then recognize that the optimization on the right hand side is exactly of the same type as the one in (\ref{eq:deftau127}) with $\b_i=1$ which is what one would get applying steps (\ref{eq:deftau11})-(\ref{eq:deftau127}) to \cite{StojnicUpper10}'s equation $(25)$. However, as shown in \cite{StojnicUpper10,StojnicCSetam09} the threshold condition \cite{StojnicUpper10}'s equation $(25)$ would provide is exactly what the upper bounds (and essentially the optimal values) of the weak thresholds are. The only difference is that one has to slightly adjust the dimensions. What are $k$, $m$, and $n$ in \cite{StojnicCSetam09,StojnicUpper10}, now are $k^{(g)}$, $m^{(g)}$, and $n^{(g)}$ where
\begin{eqnarray}
k^{(g)} & = & \frac{(\xi_n^{(l)})^2n^2}{n-m}\nonumber \\
m^{(g)} & = & m-k+k^{(g)}\nonumber \\
n^{(g)} & = & n-k+k^{(g)}.\label{eq:defkmng}
\end{eqnarray}

To give a threshold characterization for (\ref{eq:deftau12901}) we recall on the weak threshold characterization obtained in \cite{StojnicCSetam09,StojnicUpper10} for $k=\beta n$, $m=\alpha n$, and $n$ (we assume $n\rightarrow \infty$ and ignore all $\epsilon$'s from \cite{StojnicCSetam09,StojnicUpper10}).

\begin{theorem}(Weak threshold -- exact \cite{StojnicCSetam09,StojnicUpper10})
Let $A$ be an $m\times n$ matrix in (\ref{eq:system})
with i.i.d. standard normal components. Let
the unknown $\x$ in (\ref{eq:system}) be $k$-sparse. Further, let the location and signs of nonzero elements of $\x$ be arbitrarily chosen but fixed.
Let $k,m,n$ be large
and let $\alpha=\frac{m}{n}$ and $\beta_w=\frac{k}{n}$ be constants
independent of $m$ and $n$. Let $\erfinv$ be the inverse of the standard error function associated with zero-mean unit variance Gaussian random variable.  \begin{center}
\shadowbox{$
(1-\beta_w)\frac{\sqrt{\frac{2}{\pi}}e^{-(\erfinv(\frac{1-\alpha_w}{1-\beta_w}))^2}}{\alpha_w}-\sqrt{2}\erfinv (\frac{1-\alpha_w}{1-\beta_w})=0.
$}
-\vspace{-.5in}\begin{equation}
\label{eq:thmweaktheta2}
\end{equation}
\end{center}
Then:
\begin{enumerate}
\item If $\alpha>\alpha_w$ then with overwhelming probability the solution of (\ref{eq:l1}) is the $k$-sparse $\x$ from (\ref{eq:system}).
\item If $\alpha<\alpha_w$ then with overwhelming probability there will be a $k$-sparse $\x$ (from a set of $\x$'s with fixed locations and signs of nonzero components) that satisfies (\ref{eq:system}) and is \textbf{not} the solution of (\ref{eq:l1}).
    \end{enumerate}
\label{thm:thmweakthr}
\end{theorem}
\begin{proof}
The first part was established in \cite{StojnicCSetam09} and the second one was established in \cite{StojnicUpper10}. An alternative way of establishing the same set of results was also presented in \cite{StojnicEquiv10}. Of course, the weak thresholds were first computed in \cite{DonohoPol} through a different geometric approach.
\end{proof}

A combination of (\ref{eq:defkmng}) and (\ref{eq:thmweaktheta2}) then gives the following characterization of an upper bound on the sectional threshold.

\begin{theorem}(Sectional threshold -- upper bound)
Let $A$ be an $m\times n$ matrix in (\ref{eq:system})
with i.i.d. standard normal components. Let
the unknown $\x$ in (\ref{eq:system}) be $k$-sparse. Further, let the location of nonzero elements of $\x$ be arbitrarily chosen but fixed.
Let $k,m,n$ be large
and let $\alpha=\frac{m}{n}$ and $\beta_{sec}=\frac{k}{n}$ be constants
independent of $m$ and $n$. Let $\erfinv$ be the inverse of the standard error function associated with zero-mean unit variance Gaussian random variable. Further, let $\xi_{SK}$ be as in (\ref{eq:defxisk}) and let $\alpha_{sec}$ and $\beta_{sec}$ satisfy
\begin{equation}
(1-\beta_{sec})\frac{\sqrt{\frac{2}{\pi}}e^{-(\erfinv(\frac{1-\alpha_{sec}}{1-\beta_{sec}}))^2}}{\alpha_{sec}
-\beta_{sec}+\beta_{sec}(1+\xi_{SK}\sqrt{\frac{\beta_{sec}}{1-\alpha_{sec}}})^2}-\sqrt{2}\erfinv (\frac{1-\alpha_{sec}}{1-\beta_{sec}})=0.\label{eq:thmsectheta2}
\end{equation}
If $\alpha<\alpha_{sec}$ then with overwhelming probability there will be a $k$-sparse $\x$ (from the set of $\x$'s with fixed locations of nonzero components) that satisfies (\ref{eq:system}) and is \textbf{not} the solution of (\ref{eq:l1}).
\label{thm:thmweakthr}
\end{theorem}
\begin{proof}
Follows from the previous discussion, after recognizing that an upper bound on the sectional threshold of interest can be determined through a characterization of an adjusted weak threshold characterization for system with parameters $k^{(g)}$, $m^{(g)}$, and $n^{(g)}$. According to (\ref{eq:thmweaktheta2}) (and essentially to \cite{StojnicCSetam09,StojnicUpper10}) the weak threshold characterization of the problem with parameters $k^{(g)}$, $m^{(g)}$, and $n^{(g)}$ is (of course assuming $n^{(g)}\rightarrow \infty$ and $k^{(g)}$ and $m^{(g)}$ are linearly proportional to $n^{(g)}$)
\begin{eqnarray}
& & (n^{(g)}-k^{(g)})\frac{\sqrt{\frac{2}{\pi}}e^{-(\erfinv(\frac{n^{(g)}-m^{(g)}}{n^{(g)}-k^{(g)}}))^2}}{m^{(g)}}-\sqrt{2}\erfinv (\frac{n^{(g)}-m^{(g)}}{n^{(g)}-k^{(g)}})=0\nonumber \\
& \Leftrightarrow & \frac{(n^{(g)}-k^{(g)})}{n}\frac{\sqrt{\frac{2}{\pi}}e^{-(\erfinv(\frac{n^{(g)}-m^{(g)}}{n^{(g)}-k^{(g)}}))^2}}{\frac{m^{(g)}}{n}}-\sqrt{2}\erfinv (\frac{n^{(g)}-m^{(g)}}{n^{(g)}-k^{(g)}})=0.\label{eq:proofsecadjweak}
\end{eqnarray}
Using (\ref{eq:defkmng}) one then has
\begin{eqnarray}
\frac{n^{(g)}-m^{(g)}}{n^{(g)}-k^{(g)}} & = & \frac{n-m}{n-k}=\frac{1-\alpha_{sec}}{1-\beta_{sec}}\nonumber \\
\frac{n^{(g)}-k^{(g)}}{n} & = & \frac{n-k}{n}=1-\beta_{sec}\nonumber \\
\frac{m^{(g)}}{n} & = & \frac{m-k+\frac{k^2}{n-m}(\sqrt{\frac{n-m}{k}}+\xi_{SK})^2}{n}=\alpha_{sec}-\beta_{sec}+
\beta_{sec}\left (1+\xi_{SK}\sqrt{\frac{\beta_{sec}}{1-\alpha_{sec}}}\right )^2.\nonumber \\\label{eq:proofsec1}
\end{eqnarray}
Plugging (\ref{eq:proofsec1}) back in (\ref{eq:proofsecadjweak}) gives (\ref{eq:thmsectheta2}).
\end{proof}

\noindent \textbf{Remark:} Of course, there are other ways how one can establish the above given upper bound. However, we decided to present the way that we consider fairly beautiful and that, at the same time, is not that far from the original one that we discovered while proving these results.

As stated above equation (\ref{eq:thmsectheta2}) is then enough to determine an upper bound on the sectional threshold of $\ell_1$ minimization. Numerical values of the sectional threshold obtained using (\ref{eq:thmsectheta2}) are presented in Figure \ref{fig:secthrub}. We also show in Figure \ref{fig:secthrub} the lower bounds on the sectional thresholds obtained in \cite{DonohoPol,StojnicCSetam09} and in \cite{StojnicLiftStrSec13} (we refer to those from \cite{StojnicCSetam09} as the direct sectional threshold lower bounds and to those from \cite{StojnicLiftStrSec13} as the lifted sectional threshold lower bounds). As can be seen the upper bounds obtained here are obviously not the same as the lower bounds but are not that far away either.

Also, to be completely mathematically rigorous, we should add the following. Namely, to make the above theorem operational, one needs a concrete value for $\xi_{SK}$. While an exact characterization of this quantity is known it is not explicit and one typically needs to resort to a numerical computation to completely determine it. Moreover, the known methods typically approach the true value from above, whereas what we would need here is something that approaches it from below (moreover to be again completely rigorous one should say that theoretically one may really need an infinite number of numerical computations to evaluate it exactly). However, we firmly believe that the estimate we gave above is very close to the true value and can in fact already be slightly below it. Also, even if one goes one decimal further and keeps only the first three digits (which should definitely be enough to be below the true value) the changes in the resulting curve would not be visible. Essentially, for all practical purposes the light blue curve in Figure \ref{fig:secthrub} is right where it should be, it is just that we wanted to make sure that this point is also taken into account.
\begin{figure}[htb]
\centering
\centerline{\epsfig{figure=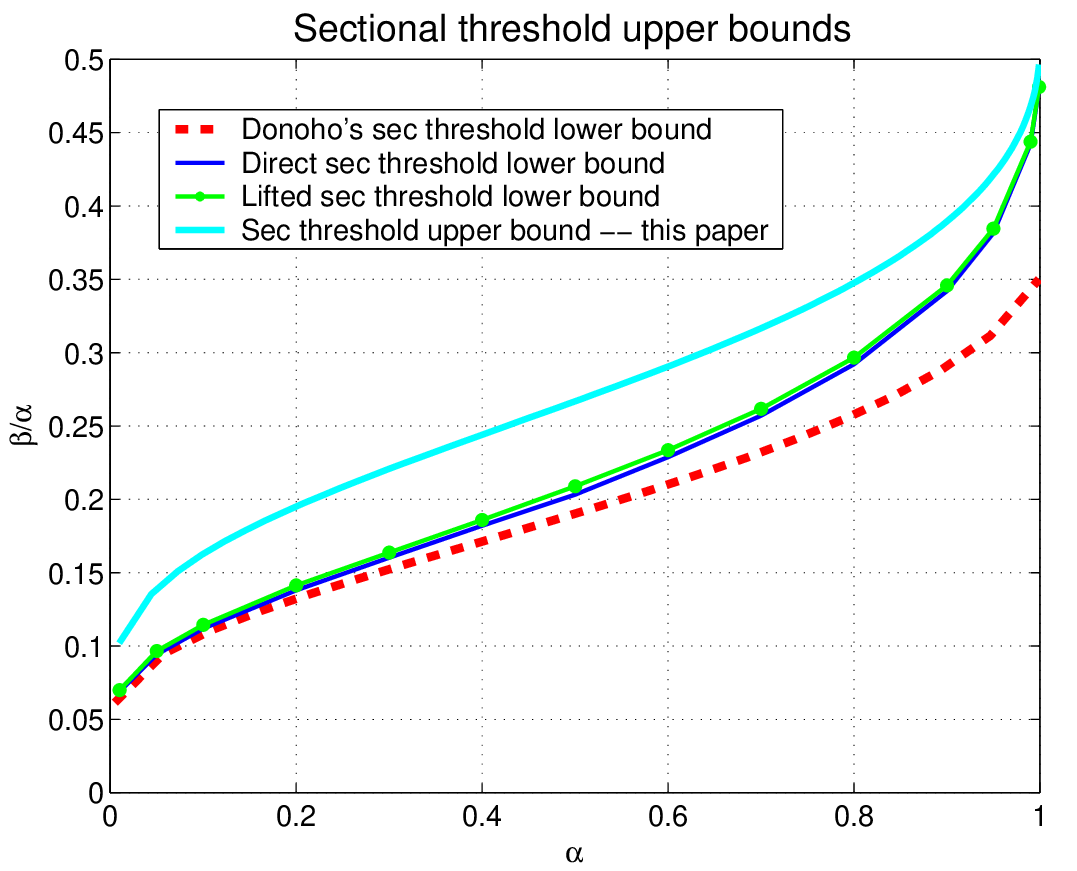,width=10.5cm,height=9cm}}
\vspace{-0.2in} \caption{\emph{Sectional} threshold, $\ell_1$-optimization --- upper bound}
\label{fig:secthrub}
\end{figure}

\section{Numerical experiments}
\label{sec:simul}

In this section we briefly discuss the results that we obtained from numerical experiments. We essentially adapted a well-known fast bit-flipping idea to design an algorithm that can numerically compute (simulate) the sectional threshold upper bounds.

\subsection{Algorithmic methodology}
\label{sec:simulalg}

Before going into the details of the obtained results we will briefly present the numerical/algorithmic methodology we used. Namely, we attempted to determine the sign of the optimal value of the objective function of the following optimization problem
\begin{eqnarray}
\min_{\w} & & -\sum_{i=n-k+1}^{n}|\w_i|+\sum_{1}^{n-k+1}|\w_i|\nonumber \\
\mbox{subject to} & & A\w=0.\label{eq:alg1}
\end{eqnarray}
Clearly negative optimum would imply that $\ell_1$ minimization sectionally fails whereas zero would mean that $\ell_1$ minimization sectionally succeeds.
Of course, this problem is not easy to solve. First there are a couple of purely numerical problems; 1) if the optimum is negative it is essentially unbounded and 2) if it is zero it is hard to believe that any finite precision machine will make it exactly zero. These problems can be handled, though. Simply adding a spherical constraint, say $\|\w\|_2\leq 50$, would fix potential unboundedness and adding a linear constraint, say $\sum_{i=1}^{n}\w_i=10$, should insure that $\w=0$ is not the solution (of course numbers $10$ and $50$ are randomly chosen; there are two things one needs to be careful about when choosing these numbers: 1) $10$ should not be small since we want to move away from zero when the optimum is nonnegative and 2) $50$ should be large enough so that a point on hyperplane $\sum_{i=1}^{n}\w_i=10$ that can potentially make objective's optimum negative is not outside the spherical constraint). Of course to insure not losing any potential solution one should resolve the problem with the same but negative linear constraint as well. Once these things are set one can look at (\ref{eq:alg1}) in the following way
\begin{eqnarray}
\tau^{(sim)}(A)=\min_{\b_i^2=1}\min_{\w} & & -\sum_{i=n-k+1}^{n}\b_i\w_i+\sum_{i=1}^{n-k+1}|\w_i|\nonumber \\
\mbox{subject to} & & A\w=0,\mbox{Sph},\pm  \mbox{Lin},\label{eq:alg2}
\end{eqnarray}
where $\mbox{Sph},\mbox{Lin}$ stand for the spherical and the linear constraint, respectively and $\pm$ indicates that the problem should be solved for both, positive and negative linear constraint. Solving the above problem over all $2^k$ different $\b$'s would produce the exact value of the optimum (in fact what we care about is the sign of the optimum). Given that $k$ can be large we instead looked at the following  simple bit-flipping algorithm. Namely, we start with $\b_i=1,n-k+1\leq i\leq n$, and with $i=n-k+1$ and then keep flipping each of $\b_i$'s (one after another, i.e. $\b_{i+1}$ after $\b_i$ for $n-k-1\leq i\leq n-1$ and $\b_{n-k+1}$ after $\b_n$) if the flipping lowers the objective value. We stop either when the objective value becomes negative or when further flipping of any of $\b_i$'s can not decrease the objective any more (or alternatively if a large number of iterations results in only marginal changes of the objective).

The above algorithm is very simple but it is far way from being the best possible (its various modifications are possible and quite often perform way better; of course quite a few different algorithms can be designed as well). Here, however, we do reemphasize that we chose it as pretty much the simplest possible while being fully aware that it is neither the most efficient complexity-wise nor the most accurate. Algorithmic studying of (\ref{eq:alg2}) is a topic on its own and since here it is not the main subject of our work we refrain from any further discussion as to how the above procedure can be improved. Instead we mention that here our goal is more to a give a rough picture/hint as to how far away from the optimum and each other our bounds are. Hence, we below present the results that we got through this simple version and leave any further consideration for a separate discussion related to algorithmic aspects of (\ref{eq:alg2}) that we will present elsewhere.

We summarize the above algorithm in Algorithm \ref{alg:alg1}. What we present in Algorithm \ref{alg:alg1} is just a sketch of the basic pseudo-code. As mentioned above one can modify it so that it stops much sooner if there are no substantial changes in the objective over a large number of iterations.

\begin{algorithm}[t]                      
\caption{A bit flipping algorithm to estimate $\mbox{sign}(\tau(A)$)}          
\label{alg:alg1}                           
\textbf{Input:} $A$, $k$, $m$, and $n$
\begin{algorithmic}[1]                    

\STATE Initialize $\tau_{min}^{(sim)}=100$, $\b_i=1$ for $n-k+1\leq i\leq n$
\STATE $j=0$, $j_{min}=1$

\WHILE{$\tau_{min}^{(sim)}\geq 0$}

\STATE $j=j+1$
\STATE $i=n-k+1+(j\mod k)$
\STATE $\b_{i}=-\b_{i}$

\STATE Solve (\ref{eq:alg2}) to obtain $\tau^{(sim)}(A)$

\IF{$(\tau^{(sim)}(A)<0)$ or $((j+1)\mod k=j_{min})$}
\STATE $\tau_{min}^{(sim)}(A)=\tau^{(sim)}(A)$
\STATE Terminate loop
\ENDIF

\IF{$\tau^{(sim)}(A)<\tau_{min}^{(sim)}(A)$}
\STATE $\tau_{min}^{(sim)}(A)=\tau^{(sim)}(A)$
\STATE $j_{min}=j$
\ELSE
\STATE $\b_{i}=-\b_{i}$
\ENDIF

\ENDWHILE

\end{algorithmic}
\textbf{Output:} $\mbox{sign}(\tau_{min}^{(sim)}(A))$
\end{algorithm}

\subsection{Numerical results}
\label{sec:simulnumres}

In all our numerical experiments we generated $m\times n$ matrices $A$ with i.i.d. zero-mean unit variance Gaussian random variables for any combination of $m$ and $n$ given in Tables \ref{tab:simulation1} and \ref{tab:simulation2}. For a fixed combination $(m,n)$ we attempted to solve underlying optimization problems for several different values of $k$ from the transition zone. For each combination $(k,m,n)$ we generated a number of different problem instances (i.e., different matrices $A$) which we call $\# \mbox{of repetitions}$ in Tables \ref{tab:simulation1} and \ref{tab:simulation2}. We then recorded the number of times our algorithm indicated that $\ell_1$ should sectionally fail, i.e. we recorded the number of times the algorithm achieved a negative objective in (\ref{eq:alg2}). All different $(k,m,n)$ combinations as well as the corresponding numbers of failed experiments are given in Tables \ref{tab:simulation1} and \ref{tab:simulation2}. Table \ref{tab:simulation1} contains the data for a range of $k$, $m$, and $n$ where $m\leq \frac{n}{2}$ or as we call it lower $\alpha$ range whereas
Table \ref{tab:simulation2} contains the data for a range of $k$, $m$, and $n$ where $m> \frac{n}{2}$ or as we call it higher $\alpha$ range.
\begin{table}[t]
\caption{ Simulation results for upper bounds of the sectional thresholds --- lower $\alpha=\frac{m}{n}\leq 0.5$ regime}\vspace{.1in}
\hspace{-.3in}\centering
\begin{tabular}{||c|c|c|c|c|c||}
\hline \hline
$n$ & $800$ & $400$ & $400$ & $400$ & $400$
\\ \hline
$m$ & $0.1n=80$ & $0.2n=80$ & $0.3n=120$ & $0.4n=160$ & $0.5n=200$ \\ \hline
$k$;  \# of errors/\# of repetitions    & $14$;  $99/100$  & $15$;  $99/100$ & $24$;  $99/100$  & $35$; $92/100$  & $50$; $99/100$    \\ \hline
$k$;  \# of errors/\# of repetitions    & $12$;  $79/100$  & $14$;  $89/100$ & $23$;  $80/100$  & $34$; $84/100$  & $48$; $90/100$    \\ \hline
$k$;  \# of errors/\# of repetitions    & $10$;  $22/100$  & $13$;  $58/100$ & $22$;  $44/100$  & $33$; $69/100$  & $46$; $53/100$    \\ \hline
$k$;  \# of errors/\# of repetitions     & $8$;   $0/100$  & $12$;  $19/100$ & $21$;  $21/100$  & $32$; $30/100$  & $44$;  $13/57$    \\ \hline
$k$;  \# of errors/\# of repetitions     & $6$;   $0/100$  & $11$;   $3/100$ & $20$;   $7/100$  & $31$; $17/100$  & $42$;  $4/100$    \\ \hline
$k$;  \# of errors/\# of repetitions     & $4$;   $0/100$  & $10$;   $0/100$ & $19$;   $1/100$  & $30$;   $1/27$  & $40$;   $0/14$    \\ \hline \hline
\end{tabular}
\label{tab:simulation1}
\end{table}

\begin{table}[t]
\caption{ Simulation results for upper bounds of the sectional thresholds --- higher $\alpha=\frac{m}{n} >0.5$ regime}\vspace{.1in}
\hspace{-.3in}\centering
\begin{tabular}{||c|c|c|c|c||}
\hline \hline
$n$ & $300$ & $200$ & $200$ & $200$
\\ \hline
$m$ & $0.6n=180$ & $0.7n=140$ & $0.8n=160$ & $0.9n=180$  \\ \hline
$k$;  \# of errors/\# of repetitions    & $51$; $100/100$  & $44$;  $98/100$ & $58$; $100/100$  & $74$; $99/100$     \\ \hline
$k$;  \# of errors/\# of repetitions    & $49$;  $95/100$  & $42$;  $81/100$ & $55$;  $92/100$  & $71$; $91/100$     \\ \hline
$k$;  \# of errors/\# of repetitions    & $47$;  $72/100$  & $40$;  $50/100$ & $53$;  $67/100$  & $69$; $68/100$     \\ \hline
$k$;  \# of errors/\# of repetitions    & $44$;   $21/99$  & $38$;  $19/100$ & $50$;  $34/100$  & $66$;  $22/57$     \\ \hline
$k$;  \# of errors/\# of repetitions    & $42$;   $8/100$  & $36$;  $13/100$ & $48$;    $5/28$  & $64$;  $9/100$     \\ \hline
$k$;  \# of errors/\# of repetitions    & $40$;   $2/100$  & $34$;    $0/31$ & $45$;   $1/100$  & $61$;  $1/100$     \\ \hline \hline
\end{tabular}
\label{tab:simulation2}
\end{table}

The interpolated data from Tables \ref{tab:simulation1} and \ref{tab:simulation2} are presented graphically in Figure \ref{fig:secthrubsim}. The color of any point in Figure \ref{fig:secthrubsim} shows the probability of having our algorithm guarantee that $\ell_1$-optimization will sectionally fail for a combination $(\alpha,\beta)$ that corresponds to that point. The colors are mapped to probabilities according to the scale on the right hand side of the figure.
The simulated results can naturally be compared to the theoretical prediction for the sectional threshold bounds. Hence, we also show in Figure \ref{fig:secthrubsim} the theoretical value for all sectional threshold bounds mentioned earlier (and shown in Figure \ref{fig:secthrub}). Since the algorithm we designed is suboptimal it may sometimes miss to find a case when $\ell_1$ should sectionally fail. That essentially means that the simulated results are also just upper bounds. Now, from Figure \ref{fig:secthrubsim} one can observe that the simulation results are exactly somewhere in between known theoretical upper and lower bounds. However, there are a couple of comments we need to add. The dimensions we simulated may not be large enough to reflect the real thresholds and at the same time we do not know how suboptimal the applied algorithm is (increasing the dimension could potentially lift the purple region while using optimal algorithms could lower it). Overall, we believe that the true thresholds are substantially closer to the green curve than to the light blue one, i.e. we believe that the lower bounds we created in \cite{StojnicCSetam09} and especially those we created in \cite{StojnicLiftStrSec13} are fairly close to the true sectional thresholds.

\begin{figure}[htb]
\centering
\centerline{\epsfig{figure=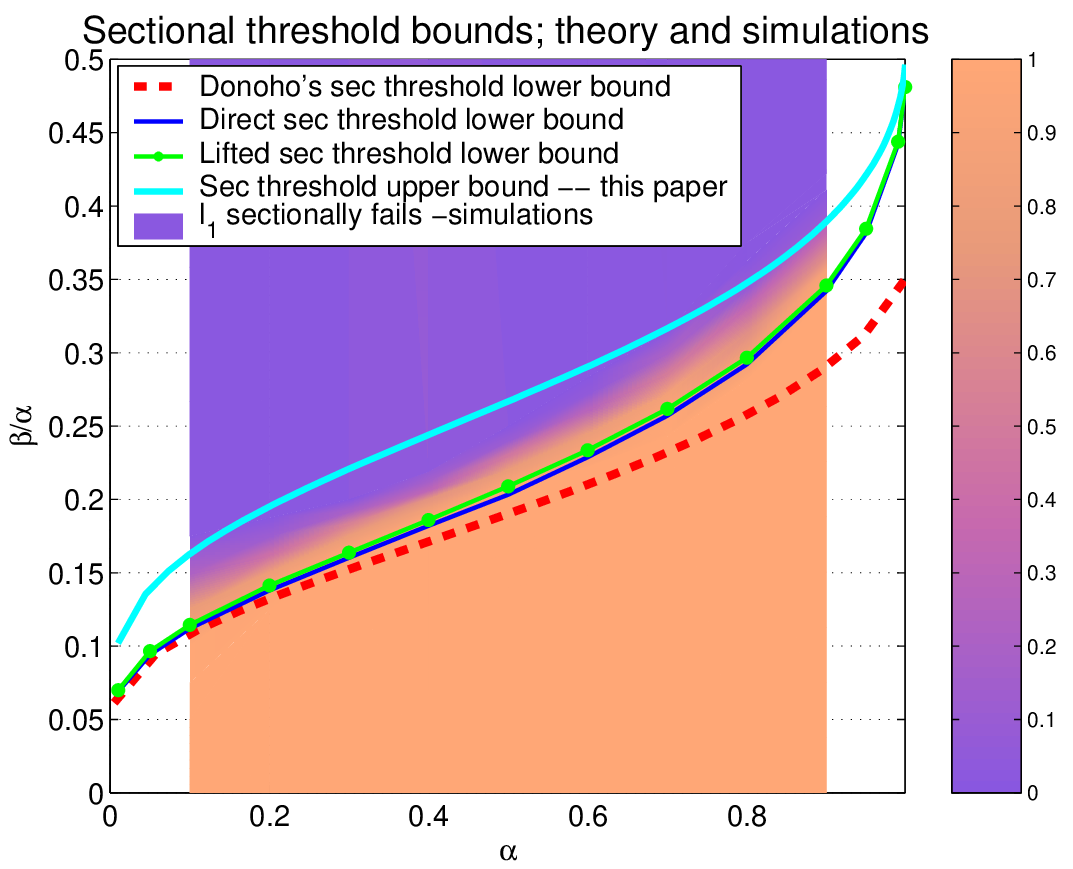,width=12cm,height=9cm}}
\vspace{-0.2in} \caption{\emph{Sectional} threshold, $\ell_1$-optimization --- upper bound; simulations}
\label{fig:secthrubsim}
\end{figure}

\section{Discussion}
\label{sec:discuss}

In this paper we considered under-determined linear systems of equations with sparse solutions.
We looked from a theoretical point of view at a classical polynomial-time
$\ell_1$-optimization algorithm. Barriers of $\ell_1$ performance are typically referred to as its thresholds. Depending if one is interested in a typical or worst case behavior one then distinguishes between the \emph{weak} thresholds that relate to a typical behavior on one side and the \emph{sectional} and \emph{strong} thresholds that relate to the worst case behavior on the other side. In this paper, we revisited the \emph{sectional} thresholds. Under the assumption that the system matrix $A$ has i.i.d. standard normal components,
we derived upper bounds on the values of the recoverable sectional
thresholds in the so-called linear regime, i.e. in the regime when
the recoverable sparsity is proportional to the length of the
unknown vector. Obtained upper bounds are relatively close to the known lower bounds we found through frameworks designed in \cite{StojnicCSetam09,StojnicLiftStrSec13}. The method we present relies on a seemingly simple but substantial progress we made in studying Hopfield models from statistical physics in \cite{StojnicHopBnds10}.

We should also mention that one can derive the upper bounds in a few different ways as well. However, we found that they typically have a more complicated presentation and don't result in a substantial improvement (i.e. while they occasionally may be better (lower) than the bounds we presented here they don't come close to matching the lower bounds). We then decided to present the method given here since in our view it is fairly elegant and in a way provides a quick assessment that the lower bounds given in \cite{StojnicCSetam09,StojnicLiftStrSec13} are highly likely not that far away from the optimal ones.

We should also mention that our results are presented for matrices $A$ with i.i.d. standard normal components. However, they hold for a way larger class of random matrices. We refrain from further discussions in this direction but instead refer to similar discussions we provided in e.g. \cite{StojnicRegRndDlt10,StojnicHopBnds10,StojnicLiftStrSec13,StojnicMoreSophHopBnds10}.

Further developments are of course possible (as is the case with pretty much any result we develop related to this and similar problems). Various specific problems that have been of interest in a broad scientific literature developed over the last few years, like quantifying the performance of $\ell_1$ type of optimization problems in solving systems with special structure of the solution vector (block-sparse, binary, box-constrained, low-rank matrix, partially known locations of nonzero components, just to name a few), systems with non-exact (noisy) solution vectors and/or equations can then have their sectional behavior bounded as well. In a few forthcoming companion papers we will present some of these applications.

What we believe is more important than adjusting the mechanism presented here to fit all problem variants is the recognition that studying the sectional thresholds may be substantially harder task than studying the corresponding weak ones. The reason is that the underlying optimization problems are combinatorial and studying their behavior (as discussed to great extent in \cite{StojnicLiftStrSec13}) typically requires a substantially larger effort.

\begin{singlespace}
\bibliographystyle{plain}
\bibliography{UpperBoundSecRefs}
\end{singlespace}

\end{document}